\documentclass[10pt,technote,final,twoside,letterpaper]{IEEEtran}

\usepackage{amsmath}   
\usepackage{amssymb}   
\usepackage{amscd}
\usepackage{latexsym}  
\usepackage{epsfig}
\usepackage{bbm}
\usepackage{revsymb4-1}

\newtheorem{defi}{Definition}
\newtheorem{lemma}[defi]{Lemma}
\newtheorem{satz}[defi]{Theorem}

\newtheorem{bem}[defi]{Remark}

\newtheorem{exempel}[defi]{Example}

\newcommand{\tr}{{\operatorname{Tr}\,}}
\newcommand{\id}{{\operatorname{id}}}
\newcommand{\supp}{{\operatorname{supp}\,}}
\newcommand{\bra}[1]{{\langle{#1}|}}
\newcommand{\ket}[1]{{|{#1}\rangle}}
\newcommand{\C}{{\mathbb{C}}}

\newcommand{\alg}[1]{{\mathfrak{#1}}}
\newcommand{\fset}[1]{{\mathcal{#1}}}

\newcommand{\1}{{\mathbbm{1}}}
\newcommand{\im}{{\operatorname{im}}}

\begin{document}

\title{Coding Theorem and Strong Converse\\ for Quantum Channels}
\author{Andreas Winter\thanks{Manuscript received September 8, 1998; revised May 4, 1999.}\thanks{The author was with SFB 343, Fakult\"at f\"ur Mathematik,
Universit\"at Bielefeld, Postfach 100131, 33501 Bielefeld, Germany.}\thanks{Current electronic address: andreas.winter@uab.cat}
}

\markboth{IEEE Transactions on Information Theory, vol. 45, no. 7, November 1999}{IEEE Transactions on Information Theory, vol. 45, no. 7, November 1999}

\maketitle


\begin{abstract}
  In this correspondence we present a new proof of Holevo's
  coding theorem for transmitting classical information through quantum
  channels, and its strong converse. The technique is largely inspired by
  Wolfowitz's combinatorial approach using types of sequences. As a by--product
  of our approach which is independent of previous ones, both in the coding
  theorem and the converse, we can give a new proof of Holevo's information
  bound.
\end{abstract}


\setcounter{page}{2481}

\section{Introduction}
\label{sec:intro}
After the recent achievements in quantum information theory,
most notably Schumacher's quantum data compression~\cite{schumacher:qucoding},
and the determination of the quantum channel capacity by
Holevo~\cite{holevo:qucapacity} (and independently by Schumacher
and Westmoreland~\cite{schumacher:capacity}), building
on ideas of Hausladen et.~al.~\cite{hausladen:qucap:purecase},
we feel that one should try to convert other and stronger techniques
of classical information theory than those used in the cited works 
to the quantum case. The present work will do this for the method
of types, as it is called by Csisz\'{a}r and
K\"orner~\cite{csiszar:koerner}, and which constitutes
a manifestly combinatorial approach to information theory,
by rephrasing it in the operator language of quantum theory
(section~\ref{sec:types:etc}). In section~\ref{sec:measurement}
we give a quantitative formulation of the intuition
that measurements with high probability
of success disturb the measured state only little.
These technical results we apply to the
coding problem for discrete memoryless quantum
channels: we give a new proof of the quantum channel coding theorem
(by a \emph{maximal code} argument, whereas previous proofs adapted
the \emph{random coding method} to quantum states),
and prove the strong converse, both in section~\ref{sec:code:bounds}.
In section~\ref{sec:holevo:bound} we demonstrate how to obtain
from these an independent, and completely
elementary proof of the Holevo bound~\cite{holevo:bound}.
We point out that our technique is also suited to the situations
of encoding under linear constraints, and with infinite
input alphabet.
\par
It should be mentioned that the strong converse results
also in recent independent work of Ogawa and
Nagaoka~\cite{ogawa:nagaoka}, by a different method.

\section{Prerequisites and notational conventions}
\label{sec:conventions}
We will use the definitions and notation of~\cite{winter:quil},
in particular  finite sets will be $\fset{A}$, $\fset{B}$,
$\fset{M}$, $\fset{X}$, ...,
quantum states (density operators) $\rho$, $\sigma$, ...,
probability distributions $P$, $Q$, ...,
and classical--quantum operations $V$, $W$, ...
(also stochastic matrices), whereas general quantum operations
(trace preserving, and completely positive maps) of C${}^*$--algebras
$\alg{X}$, $\alg{Y}$, ... are denoted as (pre-)adjoint
maps $\varphi_*$, $\psi_*$, ... .
Our algebras will be of finite dimension, and apart
from commutative ones we will confine ourselves to the C${}^*$--algebra
$\alg{L}({\cal H})$, the algebra of (bounded) linear operators of the complex
Hilbert space ${\cal H}$, even though everything works equally well for
the general case.
\par
The exponential function $\exp$ is always understood to
basis $2$, as well as the logarithm $\log$. 
The same symbol $H$ denotes the Shannon and the von Neumann entropy (as Shannon's
is the commutative case of von Neumann's).
\par
We shall need a basic fact about the trace norm $\|\cdot\|_1$
of an operator, which is the sum of the absolute values of the eigenvalues.
\par
It is known (and not difficult to prove, using the polar decomposition of
$\alpha$, see~\cite{arveson:invitation}) that
$$\|\alpha\|_1=\max_{\|B\|_\infty\leq 1} |\tr(\alpha B)|.$$ 
If $\alpha$ is selfadjoint we may write it as the difference
$\alpha_+-\alpha_-$ of its positive and its negative part. Observe that then
$$\|\alpha\|_1=\tr(\alpha_+)+\tr(\alpha_-)=\max_{-\1\leq B\leq\1} \tr(\alpha B).$$

\section{Quantum channels and codes}
\label{sec:quchannel}
The following definition is from~\cite{holevo:channels}:
a \emph{(discrete memoryless) classical--quantum channel} (cq--DMC)
is a mapping $W$ from a finite set $\fset{X}$ into
the set $\alg{S}(\alg{Y})$ of states on
the system $\alg{Y}=\alg{L}({\cal H})$, taking $x\in\fset{X}$ to $W_x$
(by linear extension we may view this as a quantum operation from
$(\C\fset{X})_*$ to $\alg{Y}_*$). For the rest of the correspondence
fix ${\cal H}$ and $\fset{X}$, $d=\dim{\cal H}$ and $a=|\fset{X}|$.
\par
From~\cite{winter:quil} we recall: for a probability distribution $P$ on $\fset{X}$ let
$PW=\sum_{x\in\fset{X}} P(x)W_x$ the average state of the channel $W$,
$H(W|P)=\sum_{x\in\fset{X}} P(x)H(W_x)$ is the conditional von Neumann entropy, 
the mutual information between a distribution and the channel is
$I(P;W)=H(PW)-H(W|P)$. Finally let $C(W)=\max_{P} I(P;W)$.
Note that these notions still make sense for infinite
$\fset{X}$ if only $W$ is required to be measurable (so $\fset{X}$ has to
carry some measurable structure): then also $H(W)$ is measurable,
and $PW$, $H(W|P)$ are expectations over the probability measure $P$.
\par
An $n$--\emph{block code} for a quantum channel $W$ is a pair $(f,{D})$,
where $f$ is a mapping from a finite
set $\fset{M}$ into $\fset{X}^n$, and ${D}$ is an observable on
$\alg{Y}^{\otimes n}$ indexed by $\fset{M}'\supset\fset{M}$,
i.e. a partition of $\1$ into positive
operators ${D}_m$, $m\in\fset{M}'$.
With the convention $W_{x^n}=W_{x_1}\otimes\cdots\otimes W_{x_n}$ for
a sequence $x^n=x_1\ldots x_n\in\fset{X}^n$
the \emph{(maximum) error probability} of the code is defined as
$$e(f,{D})=\max\{1-\tr(W_{f(m)}{D}_m): m\in\fset{M}\}.$$
We call $(f,{D})$ an $(n,\lambda)$--code, if $e(f,{D})\leq\lambda$.
Define $N(n,\lambda)$ as the maximum size $|\fset{M}|$ of an $(n,\lambda)$--code.
The \emph{rate} of an $n$--block code is defined as $\frac{1}{n}\log|\fset{M}|$.
Our main results are summarized in

\medskip
\begin{satz}
  \label{satz:capacity}
  For every $\lambda\in(0,1)$ there exists a constant $K(\lambda,a,d)$ such that
  for all cq--DMCs $W$
  $$\left|\log N(n,\lambda)-nC(W)\right|\leq K(\lambda,a,d)\sqrt{n}.$$
\end{satz}
\begin{proof}
  Combine the code construction Theorem~\ref{satz:code:construction}
  (with a probability distribution $P$ maximizing $I(P;W)$
  and $\fset{A}=\fset{X}^n$)
  and the strong converse Theorem~\ref{satz:strong:converse}.
\end{proof}

\medskip
This theorem justifies the name \emph{capacity} for the quantity
$C(W)$, even in the strong sense of Wolfowitz~\cite{wolfowitz:coding}.

\section{Typical projectors and shadows}
\label{sec:types:etc}
Let $n$ a positve integer, and consider sequences
$x^n=x_1\ldots x_n\in\fset{X}^n$. For $x\in\fset{X}$
define the counting function $N(x|x^n)=|\{i\in[n]:x_i=x\}|$.
The \emph{type} of $x^n$ is the empirical distribution
$P_{x^n}$ on $\fset{X}$ of letters $x\in\fset{X}$ in $x^n$:
$P_{x^n}(x)=\frac{1}{n}N(x|x^n)$. Obviously the number of types is
upper bounded by $(n+1)^a$; we will refer to this fact as
\emph{type counting}.

Following Wolfowitz~\cite{wolfowitz:coding} we define
\begin{equation*}\begin{split}
  \fset{T}^n_{P,\delta}=\{x^n\in\fset{X}^n:\forall x\in\fset{X}\ 
                      &|N(x|x^n)-nP(x)|\leq        \\
                      &\ \leq\delta\sqrt{n}\sqrt{P(x)(1-P(x))}\},
\end{split}\end{equation*}
the set of \emph{variance--typical sequences of approximate type $P$ with
constant $\delta\geq 0$}.
Note that $\fset{T}^n_{P,0}$ is the set of sequences of type $P$.
Defining $K=2\frac{\log e}{e}$ we have

\medskip
\begin{lemma}[Typical sequences]
  \label{lemma:wolfowitz}
  For every probability distribution $P$ on $\fset{X}$
  $$P^{\otimes n}(\fset{T}^n_{P,\delta})\geq 1-\frac{a}{\delta^2}.$$
  For $x^n\in\fset{T}^n_{P,\delta}$,
  $$|-\log P^{\otimes n}(x^n)-nH(P)|\leq Ka\delta\sqrt{n},$$
  \begin{align*}
    |\fset{T}^n_{P,\delta}| &\leq\exp\left(nH(P)+Ka\delta\sqrt{n}\right), \\
    |\fset{T}^n_{P,\delta}| &\geq\left(1-\frac{a}{\delta^2}\right)
                                    \exp\left(nH(P)-Ka\delta\sqrt{n}\right).
  \end{align*}
\end{lemma}
\begin{proof}
  See~\cite{wolfowitz:coding}. Let us only indicate the proof of the
  first inequality: 
  $\fset{T}^n_{P,\delta}$ is the intersection of $a$ events, namely
  for each $x\in\fset{X}$ that the mean of the independent Bernoulli
  variables $X_i$ with value $1$ iff $x_i=x$ has a deviation from its
  expectation $P(x)$ at most $\delta\sqrt{P(x)(1-P(x))}/\sqrt{n}$.
  By Chebyshev's inequality each of these has probability at
  least $1-1/\delta^{2}$. The rest is in fact contained in
  Lemma~\ref{lemma:variance:typical:proj} below.
\end{proof}

\medskip
The following definitions are in close analogy to this.
\par
For a state $\rho$ choose a diagonalization $\rho=\sum_j R(j)\pi_j$
and observe that the eigenvalue list $R$ is a probability distribution,
with $H(\rho)=H(R)$. Thus we may define
$$\Pi^n_{\rho,\delta}=\sum_{j^n\in\fset{T}^n_{R,\delta}}
                                \pi_{j_1}\otimes\cdots\otimes\pi_{j_n},$$
the \emph{variance--typical projector of $\rho$ with constant $\delta$}.
It is to be distinguished from the typical projector introduced by Schumacher
in~\cite{schumacher:qucoding}, which we would rather call \emph{entropy typical}.
Observe that $\Pi^n_{\rho,\delta}$ may depend on the particular
diagonalization of $\rho$. This slight abuse of notation is no harm
in the sequel, as we always fix globally diagonalizations of the
states in consideration.
\par
Let us say that an operator $B$, $0\leq B\leq\1$, is an
$\eta$--\emph{shadow} of the state $\rho$ if $\tr(\rho B)\geq\eta$.
We then have

\medskip
\begin{lemma}[Typical projector]
  \label{lemma:variance:typical:proj}
  For every state $\rho$ and integral $n$
  $$\tr(\rho^{\otimes n}\Pi^n_{\rho,\delta})\geq 1-\frac{d}{\delta^2},$$
  and with $\Pi^n=\Pi^n_{\rho,\delta}$,
  \begin{equation*}\begin{split}
    \Pi^n\exp&\left(-nH(\rho)-Kd\delta\sqrt{n}\right)
                                       \leq \Pi^n\rho^{\otimes n}\Pi^n\leq\\
             &\quad\qquad\qquad\qquad
                        \leq\Pi^n\exp\left(-nH(\rho)+Kd\delta\sqrt{n}\right),
  \end{split}\end{equation*}
  \begin{align*}
    \tr\Pi^n_{\rho,\delta} &\leq\exp\left(nH(\rho)+Kd\delta\sqrt{n}\right), \\
    \tr\Pi^n_{\rho,\delta} &\geq\left(1-\frac{d}{\delta^2}\right)
                                   \exp\left(nH(\rho)+Kd\delta\sqrt{n}\right).
  \end{align*}
  Every $\eta$--shadow $B$ of $\rho^{\otimes n}$ satisfies
  $$\tr B\geq\left(\eta-\frac{d}{\delta^2}\right)
                        \exp\left(nH(\rho)-Kd\delta\sqrt{n}\right).$$
\end{lemma}
\begin{proof}
  The first estimate is the Chebyshev inequality, as before: observe that
  $$\tr(\rho^{\otimes n}\Pi^n_{\rho,\delta})=R^{\otimes n}(\fset{T}^n_{R,\delta}).$$
  The second formula is the key: to prove it let
  $\pi^n=\pi_{j_1}\otimes\cdots\otimes\pi_{j_n}$
  one of the eigenprojections of $\rho^{\otimes n}$ contributing to
  $\Pi^n_{\rho,\delta}$. Then
  $$\tr(\rho^{\otimes n}\pi^n)=R(j_1)\cdots R(j_n)=\prod_j R(j)^{N(j|j^n)}\ .$$
  Taking logs and using the defining relation for the $N(j|j^n)$
  we find
  \begin{equation*}\begin{split}
    |\log\tr(\rho^{\otimes n}\pi^n) \!-\! nH(\rho)|
                           &=  \!\left|\sum_j \!\!-N(j|j^n)\!\log\!{R(j)} \!-\! nH(R)\right| \\
                           &\leq       \sum_j \!\!-\log{R(j)}|N(j|j^n) \!-\! nR(j)|\\
                           &\leq \sum_j -\delta\sqrt{n}\sqrt{R(j)}\log{R(j)}\\
                           &=    2\delta\sqrt{n}\sum_j -\sqrt{R(j)}\log{\sqrt{R(j)}}\\
                           &\leq 2\frac{\log e}{e}d\delta\sqrt{n}\ .
  \end{split}\end{equation*}
  The rest follows from the following Lemma~\ref{lemma:abstract:shadow:bound}.
\end{proof}

\medskip
\begin{lemma}[Shadow bound]
  \label{lemma:abstract:shadow:bound}
  Let $0\leq\Lambda\leq\1$ and $\rho$ a state commuting with $\Lambda$
  such that for some $\lambda,\mu_1,\mu_2>0$
  $$\tr(\rho\Lambda)\geq 1-\lambda\text{ and }
                  \mu_1\Lambda\leq\sqrt{\Lambda}\rho\sqrt{\Lambda}\leq\mu_2\Lambda.$$
  Then $(1-\lambda)\mu_2^{-1}\leq\tr\Lambda\leq\mu_1^{-1}$,
  and for an $\eta$--shadow $B$ of $\rho$ one has
  $\tr B\geq\left(\eta-\lambda\right)\mu_2^{-1}$.
\end{lemma}
\begin{proof}
  The bounds on $\tr\Lambda$ follow by taking traces in the inequalities
  in $\sqrt{\Lambda}\rho\sqrt{\Lambda}$ and using
  $1-\lambda\leq\tr(\rho\Lambda)\leq 1$.
  For the $\eta$--shadow $B$ observe
  \begin{equation*}\begin{split}
    \mu_2\tr B
       &\geq \tr\left(\mu_2\Lambda B\right)
           \geq \tr\left(\sqrt{\Lambda}\rho\sqrt{\Lambda} B\right) \\
       &=\tr(\rho B)-\tr\left(\left(\rho-\sqrt{\Lambda}\rho\sqrt{\Lambda}\right)B\right) \\
       &\geq \eta-\left\|\rho-\sqrt{\Lambda}\rho\sqrt{\Lambda}\right\|_1 .
  \end{split}\end{equation*}
  Since the trace norm can obviously be estimated
  by $\lambda$ we are done.
\end{proof}

\medskip
Fix now diagonalizations
$W_x=\sum_{j} W(j|x)\pi_{xj}$ (where $W(\cdot|\cdot)$ is a stochastic matrix,
the double meaning of $W$ should be no serious ambiguity).
Then define the \emph{conditional variance--typical projector of $W$ given $x^n$
with constant $\delta$} to be
$$\Pi^n_{W,\delta}(x^n)=\bigotimes_{x\in\fset{X}} \Pi^{I_x}_{W_x,\delta},$$
where $I_x=\{i\in[n]:x_i=x\}$. With the convention
$W_{x^n}=W_{x_1}\otimes\cdots\otimes W_{x_n}$ we now have

\medskip
\begin{lemma}[Conditional typical projector]
  \label{lemma:conditional:variance:typical:proj}
  For all $x^n\in\fset{X}^n$ of type $P$
  $$\tr(W_{x^n}\Pi^n_{W,\delta}(x^n))\geq 1-\frac{ad}{\delta^2}\ ,$$
  and with $\Pi^n=\Pi^n_{W,\delta}(x^n)$
  \begin{equation*}\begin{split}
    \Pi^n\exp&\left(-nH(W|P)-Kd\sqrt{a}\delta\sqrt{n}\right)
                                         \leq\Pi^n W_{x^n}\Pi^n\leq \\
             &\qquad\qquad\leq \Pi^n\exp\left(-nH({W}|P)+Kd\sqrt{a}\delta\sqrt{n}\right),
  \end{split}\end{equation*}
  \begin{align*}
    \tr\Pi^n_{W,\delta}(x^n) &\leq\exp\left(nH(W|P)+Kd\sqrt{a}\delta\sqrt{n}\right),\\
    \tr\Pi^n_{W,\delta}(x^n) &\geq\left(1-\frac{ad}{\delta^2}\right)
                                       \exp\left(nH(W|P)-Kd\sqrt{a}\delta\sqrt{n}\right).
  \end{align*}
  Every $\eta$--shadow $B$ of $W_{x^n}$ satisfies
  $$\tr B\geq\left(\eta-\frac{ad}{\delta^2}\right)
                                    \exp\left(nH({W}|P)-Kd\sqrt{a}\delta\sqrt{n}\right).$$
\end{lemma}
\begin{proof}
  The first estimate follows simply by applying
  Lemma~\ref{lemma:variance:typical:proj} $a$ times,
  the second formula is by piecing together the corresponding
  formulae from Lemma~\ref{lemma:variance:typical:proj},
  using $\sum_{x\in\fset{X}}\sqrt{P(x)}\leq\sqrt{a}$.
  The rest is by the shadow bound Lemma~\ref{lemma:abstract:shadow:bound}.
\end{proof}

\medskip
We need a last result on the behaviour of $W_{x^n}$
under a typical projector:
\begin{lemma}[Weak law of large numbers]
  \label{lemma:weak:law}
  Let $x^n\in\fset{X}^n$ of type $P$. Then
  $$\tr(W_{x^n}\Pi^n_{PW,\delta\sqrt{a}})\geq 1-\frac{ad}{\delta^2}.$$
\end{lemma}
\begin{proof}
  Diagonalize $PW=\sum_j q_j\pi_j$, and let the quantum operation
  $\kappa_*:\alg{L}({\cal H})_*\rightarrow\alg{L}({\cal H})_*$
  be defined by $\kappa_*(\sigma)=\sum_j \pi_j\sigma\pi_j$.
  We claim that
  $$\Pi^n_{PW,\delta\sqrt{a}}\geq \Pi^n_{\kappa_*W,\delta}(x^n).$$
  Indeed let $\pi^n=\pi_{j_1}\otimes\cdots\otimes\pi_{j_n}$
  one of the product states constituting
  $\bigotimes_{x\in\fset{X}} \Pi^{I_x}_{\kappa_*W_x,\delta}$, i.e.
  with $\kappa_*W_x=\sum_j q_{j|x}\pi_j$,
  $$\forall x\!\in\!\fset{X}\,\forall j\ \left| N(j|j^{I_x})-q_{j|x}|I_x| \right|\leq
                                           \delta\sqrt{|I_x|}\sqrt{q_{j|x}(1-q_{j|x})}\ .$$
  Hence (using $|I_x|=P(x)n$)
  \begin{equation*}\begin{split}
    | N(j|j^n)-q_jn | &\leq\sum_{x\in\fset{X}} \left| N(j|j^{I_x})-q_{j|x}|I_x| \right| \\
               &\leq\sum_{x\in\fset{X}} \delta\sqrt{n}\sqrt{P(x)}\sqrt{q_{j|x}(1-q_{j|x})}\\
               &\leq\delta\sqrt{n}\sqrt{a}\sqrt{\sum_{x\in\fset{X}} P(x)q_{j|x}(1-q_{j|x})}\\
               &\leq\delta\sqrt{n}\sqrt{a}\sqrt{q_j(1-q_j)}\ ,
  \end{split}\end{equation*}
  the last inequality by concavity of the map $x\mapsto x(1-x)$,
  and $q_j=\sum_{x\in\fset{X}} P(x)q_{j|x}$. Hence $\pi^n$ occurs
  in the sum for $\Pi^n_{PW,\delta\sqrt{a}}$, and our claim is proved.
  \par
  Thus we can estimate
  \begin{equation*}\begin{split}
    \tr(W_{x^n}\Pi^n_{PW,\delta\sqrt{a}})
            &=\tr\left((\kappa_*^{\otimes n}W_{x^n})\Pi^n_{PW,\delta\sqrt{a}}\right)\\
            &\geq \tr\left((\kappa_*^{\otimes n}W_{x^n})\Pi^n_{\kappa_*W,\delta}(x^n)\right)\\
            &\geq 1-\frac{ad}{\delta^2}\ ,
  \end{split}\end{equation*}
  the last line by Lemma~\ref{lemma:conditional:variance:typical:proj}.
\end{proof}

\section{On good measurements}
\label{sec:measurement}
We start with a short consideration of \emph{fidelity}:
\par
Assume in the following that $\rho$ is a pure state, $\sigma$ may be mixed.
We want to compare the trace norm distance
$D(\rho,\sigma)=\frac{1}{2}\|\rho-\sigma\|_1$, and
the (pure state) fidelity $F(\rho,\sigma)=\tr(\rho\sigma)$.

\medskip
\begin{lemma}[Pure state]
  \label{lemma:pure:state}
  Let $\rho=\ket{\psi}\bra{\psi}$ and $\sigma=\ket{\phi}\bra{\phi}$ pure states. Then
  $$1-F(\rho,\sigma)=D(\rho,\sigma)^2\ .$$
\end{lemma}
\begin{proof}
  We may assume $\ket{\psi}=\alpha\ket{0}+\beta\ket{1}$ and
  $\ket{\phi}=\alpha\ket{0}-\beta\ket{1}$, with $|\alpha|^2+|\beta|^2=1$.
  A straightforward calculation shows
  $F=(|\alpha|^2-|\beta|^2)^2$, and $D=2|\alpha\beta|$. Now
  \begin{equation*}\begin{split}
    1-F &=1-(|\alpha|^2-|\beta|^2)^2\\
        &=(1+|\alpha|^2-|\beta|^2)(1-|\alpha|^2+|\beta|^2)\\
        &=4|\alpha\beta|^2=D^2\ .
  \end{split}\end{equation*}
\end{proof}

\medskip
\begin{lemma}[Mixed state]
  \label{lemma:mixed:state}
  Let $\sigma$ an arbitrary mixed state (and $\rho$ pure as above). Then
  $$D\geq 1-F\geq D^2\ .$$
\end{lemma}
\begin{proof}
  Write $\sigma=\sum_j q_j\pi_j$ with pure states $\pi_j$. Then
  \begin{equation*}\begin{split}
    1-F(\rho,\sigma) &=\sum_j q_j\left(1-F(\rho,\pi_j)\right)=\sum_j q_jD(\rho,\pi_j)^2\\
                     &\geq \left(\sum_j q_jD(\rho,\pi_j)\right)^{\!\! 2}
                        \geq D(\rho,\sigma)^2\ .
  \end{split}\end{equation*}
  Conversely: extend $\rho$ to the observable $(\rho,\1-\rho)$
  and consider the quantum operation
  $$\kappa_*:\sigma\longmapsto \rho\sigma\rho+(\1-\rho)\sigma(\1-\rho).$$
  Then (with monotonicity of $\|\cdot\|_1$ under quantum operations)
  $$2D=\|\rho-\sigma\|_1\geq\|\kappa_*\rho-\kappa_*\sigma\|_1=\|\rho-\kappa_*\sigma\|_1$$
  (since $\rho=\kappa_*\rho$). Hence with $F=\tr(\sigma\rho)$
  \begin{equation*}\begin{split}
    2D &\geq\big{\|}(1-F)\rho-\tr(\sigma(\1-\rho))\pi\big{\|}_1\\
       &=(1-F)+(1-F)=2(1-F),
  \end{split}\end{equation*}
  for a state $\pi$ supported in $\1-\rho$.
\end{proof}

\medskip
Observe that the inequalities of this lemma still hold if only $\sum_j q_j\leq 1$.
\par
Now we are ready to prove the main object of the present section:

\medskip
\begin{lemma}[Gentle measurement]
  \label{lemma:tender:operator}
  Let $\rho$ be a state, and $X$ a positive operator with $X\leq\1$
  and $1-\tr(\rho X)\leq\lambda\leq 1$. Then
  $$\left\|\rho-\sqrt{X}\rho\sqrt{X}\right\|_1\leq\sqrt{8\lambda}\ .$$
\end{lemma}
\begin{proof}
  Let $Y=\sqrt{X}$ and write $\rho=\sum_k p_k\pi_k$ with one--dimensional
  projectors $\pi_k$ and weights $p_k\geq 0$. Now
  \begin{equation*}\begin{split}
    \|\rho-Y\rho Y\|_1^2 &\leq \left(\sum_k p_k\|\pi_k-Y\pi_k Y\|_1\right)^2\\
                         &\leq \sum_k p_k\|\pi_k-Y\pi_k Y\|_1^2\\
                         &\leq 4\sum_k p_k(1-\tr(\pi_k Y\pi_k Y))\\
                         &\leq 8\sum_k p_k(1-\tr(\pi_k Y))\\
                         &=    8(1-\tr(\rho Y))\\
                         &\leq 8(1-\tr(\rho X))\leq 8\lambda
  \end{split}\end{equation*}
  by triangle inequality, convexity of $x\mapsto x^2$,
  Lemma~\ref{lemma:mixed:state}, $1-x^2\leq 2(1-x)$, and $X\leq Y$.
\end{proof}

\section{Code bounds}
\label{sec:code:bounds}
We can now give a new proof of the quantum channel coding theorem by
a maximal code argument (which in the classical case is due to
Feinstein~\cite{feinstein:maximal}), and prove the strong converse.

\medskip
\begin{satz}[Code construction]
  \label{satz:code:construction}
  For $\lambda,\tau\in(0,1)$ there exist $\delta>0$ and
  a constant $K(\lambda,\tau,a,d)$ such that for every
  cq--DMC $W$, probability distribution $P$ on $\fset{X}$,
  $n>0$, and $\fset{A}\subset\fset{X}^n$ with
  $P^{\otimes n}(\fset{A})\geq\tau$ there is an $(n,\lambda)$--code
  $(f,D)$ with the properties
  $$\forall m\in\fset{M}\quad f(m)\in\fset{A}\text{ and }
                              \tr D_m\leq\tr\Pi^n_{W,\delta}(f(m)),$$
  and $|\fset{M}|\geq\exp\left(nI(P;W)-K(\lambda,\tau,a,d)\sqrt{n}\right)$.
\end{satz}
\begin{proof}
  Let $\fset{A}'=\fset{A}\cap\fset{T}^n_{P,\sqrt{{2ad}/{\tau}}}$
  (thus $P^{\otimes n}(\fset{A}')\geq\tau/2$) and
  $(f,D)$ a maximal (i.e. non--extendible) $(n,\lambda)$--code with
  $$\forall m\in\fset{M}\quad f(m)\in\fset{A}'\text{ and }
                              \tr D_m\leq\tr\Pi^n_{W,\delta}(f(m)),$$
  where $\delta=\sqrt{\frac{2ad}{\lambda}}$.
  In particular (by Lemma~\ref{lemma:conditional:variance:typical:proj})
  $$\tr D_m\leq\exp\!\left(\! nH(W|P)\!
        +\!(Kd\sqrt{a}\delta+Ka\sqrt{\frac{2ad}{\tau}}\log d)\sqrt{n}\!\right)\!\!.$$
  Of course $\fset{M}$ may be empty. We claim
  however that $B=\sum_{m\in\fset{M}} D_m$ is an $\eta$--shadow
  for all $W_{x^n}$, $x^n\in\fset{A}'$,
  with $\eta=\min\{1-\lambda,\lambda^2/32\}$.
  \par
  This is clear for codewords, and for other $x^n$ we could else
  extend $(f,D)$ with the codeword $x^n$ and corresponding
  observable operator
  $$D_{x^n}=\sqrt{\1-B}\Pi^n_{W,\delta}(x^n)\sqrt{\1-B}.$$
  To see this note first that $D_{x^n}\leq\1-B$, and
  $\tr D_m\leq\tr\Pi^n_{W,\delta}(x^n)$.
  Now apply Lemma~\ref{lemma:tender:operator} to the assumption
  $\tr\left(W_{x^n}(\1-B)\right)\geq 1-\lambda^2/32$ and obtain
  $$\left\|W_{x^n}-\sqrt{\1-B}W_{x^n}\sqrt{\1-B}\right\|_1\leq\frac{\lambda}{2}\ .$$
  Hence we can estimate (with $\Pi^n=\Pi^n_{W,\delta}(x^n)$):
  \begin{equation*}\begin{split}
    \tr(W_{x^n}D_{x^n}) &=\tr(W_{x^n}\Pi^n)- \\
               &\quad-\tr\left((W_{x^n}-\sqrt{\1-B}W_{x^n}\sqrt{\1-B})\Pi^n\right) \\
               &\geq 1-\frac{\lambda}{2}
                      -\left\|W_{x^n}-\sqrt{\1-B}W_{x^n}\sqrt{\1-B}\right\|_1 \\
               &\geq 1-\lambda.
  \end{split}\end{equation*}
  This proves our claim, and averaging over $P^{\otimes n}$ we find
  $$\tr\left((PW)^{\otimes n}B\right)\geq \eta\tau/2,$$
  from which, by Lemma~\ref{lemma:variance:typical:proj}, we deduce
  \begin{equation*}\begin{split}
    \sum_{m\in\fset{M}} \tr D_m &=   \tr B \\
                                &\geq\left(\frac{\eta\tau}{2}-\frac{d}{\delta_0^2}\right)
                                 \exp\left(nH(PW)-Kd\delta_0\sqrt{n}\right).
  \end{split}\end{equation*}
  Choosing $\delta_0=\sqrt{\frac{4d}{\eta\tau}}$ the proof is complete.
\end{proof}

\medskip
\begin{bem}
  \label{bem:code:structure}
  It is interesting to note from the proof that the decoder may be chosen a
  von Neumann observable (i.e. all its operators are mutually orthogonal
  projectors). This is because if $(f,D)$ is of this type, then $B$ is a projector,
  and this means that we may instead of the constructed $D_{x^n}\leq \1-B$ use
  the projector $D_{x^n}'=\supp D_{x^n}$: this is still bounded
  by $\1-B$, only decreases the error probability, and obeys the size
  condition:
  $\tr\supp D_{x^n}=\dim\im\, D_{x^n}\leq
             \dim\im\,\Pi^n_{W,\delta}(x^n)=\tr\Pi^n_{W,\delta}(x^n)$.
  \par
  On the other hand
  it would be nice if we could decide if the decoder may consist
  of separable operators. It is clear that a product observable cannot do,
  as was pointed out by Holevo~\cite{holevo:overview}: otherwise larger
  capacities could not be reached using block decoding. But it may be that
  nonlocality as in the recent work of Bennett et.~al.~\cite{bennett:nonlocal}
  is sufficient, and genuine entanglement is not needed
  (as was proposed in the cited work of Holevo).
\end{bem}

\medskip
\begin{bem}
  Our method of proof might seem very abstract. In fact
  it is not, as the argument in the proof may be understood as a
  \emph{greedy method} of extending a given code: start from the empty
  code, and add codewords after the prescription of the proof, until
  you are stuck. The theorem then guarantees that the resulting code is
  rather large.
\end{bem}

\medskip
\begin{satz}[Strong converse]
  \label{satz:strong:converse}
  For $\lambda\in(0,1)$ there exists a constant $K(\lambda,a,d)$
  such that for every cq--DMC $W$ and $(n,\lambda)$--code $(f,D)$
  $$|\fset{M}|\leq\exp\left(nC(W)+K(\lambda,a,d)\sqrt{n}\right).$$
\end{satz}
\begin{proof}
  We will prove even a little more: if additionally all codewords
  are of the same type $P$ then
  $$|\fset{M}|\leq\frac{4}{1-\lambda}
                  \exp\left(nI(P;W)+2Kd\sqrt{a}\delta\sqrt{n}\right),$$
  with $\delta=\frac{\sqrt{32ad}}{1-\lambda}$, which by type counting
  implies the theorem.
  \par
  To prove this modify the code as follows: construct new decoding
  operators
  $$D_m'=\Pi^n_{PW,\delta\sqrt{a}}D_m\Pi^n_{PW,\delta\sqrt{a}}.$$
  Then $(f,D')$ is an $(n,\frac{1+\lambda}{2})$--code because for
  $m\in\fset{M}$, with $\Pi^n=\Pi^n_{PW,\delta\sqrt{a}}$
  \begin{equation*}\begin{split}
    \tr(W_{f(m)}D_m') &=\tr(W_{f(m)}D_m)- \\
                      &\quad-\tr\left((W_{f(m)}-\Pi^n W_{f(m)}\Pi^n)D_m\right) \\
                      &\geq 1-\lambda
                        -\|W_{f(m)}-\Pi^n W_{f(m)}\Pi^n\|_1                \\
                      &\geq 1-\lambda-\sqrt{\frac{8ad}{\delta^2}}
                       =   \frac{1-\lambda}{2}\ .
  \end{split}\end{equation*}
  Now by Lemma~\ref{lemma:conditional:variance:typical:proj}
  $$\tr D_m'\geq\frac{1-\lambda}{4}\exp\left(nH(W|P)-Kd\sqrt{a}\delta\sqrt{n}\right).$$
  On the other hand (with Lemma~\ref{lemma:variance:typical:proj})
  \begin{equation*}\begin{split}
    \sum_{m\in\fset{M}}\tr D_m' &\leq\tr\Pi^n_{PW,\delta\sqrt{a}} \\
                                &\leq\exp\left(nH(PW)+Kd\sqrt{a}\delta\sqrt{n}\right),
  \end{split}\end{equation*}
  and we are done.
\end{proof}

\section{Holevo bound}
\label{sec:holevo:bound}
An interesting application of our converse
Theorem~\ref{satz:strong:converse}
is in a new, and completely elementary proof of
the famous Holevo bound:
\par
For a cq--DMC $W:\fset{X}\rightarrow\alg{S}(\alg{L}({\cal H}))$,
a probability distribution $P$ on $\fset{X}$, and an observable $D$ on
$\alg{Y}$, say indexed
by $\fset{Y}$, the composition $D_*\circ W:\fset{X}\rightarrow\fset{Y}$
is a \emph{classical} channel.
\par
Holevo in~\cite{holevo:bound} considers
$C_1=\max_{P,D} I(P;D_*\circ W)$
(the capacity if one is restricted to tensor product observables!)
and proves analytically $C_1\leq C(W)$. For us this is now clear, since all codes
for the classical channels $D_*\circ W$ (whose maximal rates are asymptotically
bounded by $C_1$) can be interpreted as special channel codes for $W$.
\par
But we can show even a little more, namely Holevo's original
\emph{information bound} $I(P;D_*\circ W)\leq I(P;W)$, from which the
capacity estimate clearly follows.

\medskip
\begin{proof}
  Assume the opposite, $I(P;D_*\circ W)>I(P;W)$. Then by the well known
  classical coding theorem (alternatively the quantum channel
  coding Theorem~\ref{satz:code:construction} which
  generalizes the classical case) there is to every
  $\delta>0$ an infinite sequence of $(n,1/2)$--codes with codewords
  chosen from
  $\fset{T}^n_{P,\sqrt{2a}}$ for the channel $D_*\circ W$ with
  rates exceeding $I(P;D_*\circ W)-\delta$. Restricting to a single
  type of codewords we find constant composition codes (of type $P_n$)
  with rate exceeding $I(P;D_*\circ W)-2\delta$ (if $n$ is large enough).
  \par
  As already explained these are special channel codes for $W$, so
  by Theorem~\ref{satz:strong:converse} (proof) their rates are upper
  bounded by $I(P_n;W)+\delta$ (again, $n$ large enough), hence
  $$I(P;D_*\circ W)-2\delta \leq I(P_n;W)+\delta.$$
  Collecting inequalities we find
  $$I(P;W) < I(P;D_*\circ W) \leq I(P_n;W)+3\delta.$$
  But since $P_n\rightarrow P$ by assumption and by the continuity
  of $I$ in $P$, since furthermore $\delta$ is arbitrarily small, we end up with
  $$I(P;W) < I(P;D_*\circ W) \leq I(P;W),$$
  a contradiction.
\end{proof}

\section{Conclusion} 
\label{sec:final}
We proved the quantum channel coding theorem and its strong converse
by methods new to quantum information theory (but which are very
close to established methods in classical information theory),
and showed how to obtain the Holevo bound as a corollary.
\par
We want to point out that our technique for proving the code bounds
yields also the coding theorem and strong converse under
\emph{linear constraints}
(see Holevo~\cite{holevo:constraints} for definitions and capacity
formula): simply because satisfying the linear constraints is a property of
whole types, not just individual sequences.
\par
Also we can prove the coding theorem and strong converse in the
case of arbitrary (product) signal states in a general (discrete memoryless)
\emph{quantum--quantum channel} (qq--DMC), see~\cite{holevo:channels}:
this is a completely positive and unit preserving map
$\varphi:\alg{A}_2\rightarrow\alg{A}_1$ between finite dimensional
C${}^*$--algebras $\alg{A}_1,\alg{A}_2$ (or rather its state map
$\varphi_*:\alg{S}(\alg{A}_1)\rightarrow\alg{S}(\alg{A}_2)$).
This includes the cq--DMC as the special case $\alg{A}_1=\C\fset{X}$ and
$\alg{A}_2=\alg{L}({\cal H})$.
\par
An $(n,\lambda)$--code for this channel is a pair $(F,D)$ with a map
$F:\fset{M}\rightarrow\alg{S}(\alg{A}_1^{\otimes n})$ and an observable $D$
on $\alg{A}_2^{\otimes n}$ indexed
by $\fset{M}'\supset\fset{M}$, such that the error probability
$$e(F,D)=\max\{1-\tr(\varphi_*^{\otimes n}(F(m))\cdot D_m):m\in\fset{M}\}$$
is at most $\lambda$. We will consider only the case that the $F(m)$ are
product states (such codes we call \emph{1--separable}, 
following~\cite{schumacher:capacity}, and the corresponding operational
capacity \emph{product state capacity} $C^{(1)}({\varphi_*},\lambda)$),
so the channel and all its possible codes are determined
by the image of $\varphi_*$ in $\alg{S}(\alg{A}_2)$, a compact convex set.
Thus we are back in our original situation, with $\alg{W}$ now
a compact convex set of states in $\alg{L}({\cal H})$ and $W=\id_\alg{W}$.
The capacity we denote $C(\alg{W},\lambda)$, it was for $\lambda\rightarrow 0$
determined by Schumacher and Westmoreland~\cite{schumacher:capacity} (an improved
argument for their weak converse may be found in~\cite{winter:quil}).
\par
With the methods presented in this correspondence one can prove

\medskip
\begin{satz}
  \label{satz:qq:1:capacity}
  With the above notations and $\lambda\in(0,1)$:
  $$C(\alg{W},\lambda)=\sup_{\text{finite }\fset{W}\subset\alg{W}} C(\fset{W}).$$
  Furthermore the supremum
  is in fact a maximum, which is assumed by a finite $\fset{W}\subset\alg{W}$
  of cardinality at most $\dim_{\C}\alg{A}_2$ and consisting of
  extremal points of $\alg{W}$.
\end{satz}
\par
The proof of the capacity formula is given in full
in~\cite{winter:diss}. The second part of the statement
is from~\cite{fujiwara:nagaoka}.

\section*{Acknowledgements}
Thanks to Peter L\"ober for discussions on various aspects of the present work,
and to Prof.~Rudolf Ahlswede for his teaching of information theory.
I am indebted to the referee, whose remarks contributed much to
the clarity of the above presentation.

\thispagestyle{empty}

\nocite{*}
\bibliographystyle{IEEE}

\end{document}